\documentclass[a4paper]{article}

\usepackage[a4paper, margin=1in]{geometry}
\usepackage{comment}
\usepackage{bm}
\usepackage{multirow}
\usepackage{balance}
\usepackage{cite}
\usepackage{amsmath,amssymb,amsfonts,amsthm}
\usepackage[linesnumbered,ruled,vlined]{algorithm2e}
\usepackage{graphicx}
\usepackage{subfig}
\usepackage{textcomp}
\usepackage{xcolor}
\def\BibTeX{{\rm B\kern-.05em{\sc i\kern-.025em b}\kern-.08em
    T\kern-.1667em\lower.7ex\hbox{E}\kern-.125emX}}
    
\SetKwInput{KwInput}{Input}
\SetKwInput{KwOutput}{Output}
\newcommand{\N}{\mathbb{N}}
\newcommand{\HD}{\mathcal{HD}}
\newcommand{\bigO}{\mathcal{O}}
\newcommand{\myFor}[2]{\textbf{for}~#1~\textbf{do}~#2}

\newcommand{\SV}{{\cal S}}
\newcommand{\CP}{{\cal CP}}
\newcommand{\DB}{{\sf DB}}
\newcommand{\C}{{\cal C}}
\newcommand{\HDbound}{\overline{\HD}}
\newcommand{\system}{{\em LightPIR\xspace}}

\newtheorem{lemma}{Lemma}
    
\title{\emph{LightPIR}: Privacy-Preserving Route Discovery\\for Payment Channel Networks
}

\author{
Krzysztof Pietrzak$^\ddag$\footnote{This project has received funding from the European Research Council (ERC) under the European Union’s Horizon 2020 research and innovation programme
(682815 - TOCNeT)}\\
\texttt{krzysztof.pietrzak@ist.ac.at}
\and
Iosif Salem$^\S$\\
\texttt{iosif.salem@univie.ac.at}
\and
Stefan Schmid$^\S$\footnote{This project has received funding from the European Research Council (ERC) under the European Union’s Horizon 2020 research and innovation programme
(864228 - AdjustNet)}\\
\texttt{stefan\_schmid@univie.ac.at}
\and
Michelle Yeo$^\ddag$\\
\texttt{michelle.yeo@ist.ac.at}
\and
$^\ddag$IST Austria\\
$^\S$Faculty of Computer Science, 
University of Vienna
}

\date{}

\begin{document}

\maketitle

\begin{abstract}
Payment channel networks are a promising approach to
improve the scalability of cryptocurrencies:
they allow to perform transactions 
in a peer-to-peer fashion, along multi-hop 
routes in the network, without requiring consensus on the
blockchain. 
However, during the discovery of cost-efficient routes for the transaction,
critical information may be revealed about the transacting entities.

This paper initiates the study of privacy-preserving route discovery
mechanisms for payment channel networks. In particular, we present
\system, an approach which allows a source to efficiently discover a shortest path to its
destination without revealing any information about the endpoints of the transaction.
The two main observations which allow for an efficient solution
in \system~are that: (1) surprisingly, hub labelling algorithms -- which were developed to preprocess ``street network like" graphs so one can later efficiently compute shortest paths -- also work well for the graphs underlying payment channel networks, and that (2) hub labelling algorithms can be directly combined with private information retrieval.

\system \space relies on a simple hub labeling heuristic on top of existing hub labeling algorithms which leverages the specific topological features of cryptocurrency networks to further minimize storage and bandwidth overheads.
In a case study considering the Lightning network, 
we show that our approach is an order of magnitude more efficient compared
to a privacy-preserving baseline based on using private information retrieval on a database that stores all pairs shortest paths. 
\end{abstract}


\section{Introduction}

As the popularity of cryptocurrencies is growing explosively,
the inefficiency of the underlying consensus protocol increasingly becomes
a bottleneck, preventing fast payments at scale. 
A promising solution to this scalability challenge 
are emerging payment channel networks~\cite{gudgeon2019sok,csur21crypto} such as 
the Lightning network~\cite{poon2016bitcoin}, which allow to perform transactions off-chain,
without requiring consensus on the
blockchain. In a nutshell, a payment channel is a cryptocurrency transaction which escrows or dedicates money on the
blockchain for exchange with a given user and duration. 
Users
can also interact if they do not share a direct payment channel:
they can route transactions through intermediaries.
To incentivize nodes (the
intermediaries) to contribute to the transaction routing, 
payment channel networks typically allow 
intermediaries to charge nodes which route through 
them a nominal fee. This is also the
approach taken in the Lightning network which serves us as
a case study in this paper.

Payment channel networks however also introduce new challenges.
In particular, since different routes come at different fees, nodes require
scalable mechanisms to discover ``short” (i.e., low-cost) routes.
To ensure scalability and support lightweight nodes (e.g., wallets) in the network, 
these route discovery mechanisms must be efficient, both in terms
of disk space and control traffic. Especially since payment channel
networks can be more dynamic compared to classic communication networks
(e.g., channels and liquidities can change over time)~\cite{gudgeon2019sok},
it is crucial that resource-constrained
nodes do not have to store and maintain the entire network topology
to be able to compute a route toward the
transaction’s target.
Thus, over the last years, several innovative approaches 
for a more scalable routing in payment channel networks 
have been proposed~\cite{prihodko2016flare,roos2017settling,trampoline,
malavolta2017silentwhispers,sivaraman2018routing,aft20}.
For example, Lightning introduced the notion of trampoline nodes~\cite{trampoline}
to provide route discovery services to light-weight clients.
The trampoline nodes are more powerful nodes which know the entire topology (e.g., using
gossiping); a wallet then just needs to know how to reach the route server nodes
in its neighborhood and can then request the desired route.

Besides scalability, payment channel networks also need to
provide a high degree of privacy: it is critical that other nodes in the payment channel
networks cannot learn who transacts with whom.  In order to 
provide anonymity in the transaction routing, 
payment channel networks typically 
implement onion routing~\cite{dingledine2004tor,roos2017settling,
malavolta2017silentwhispers}.
However, while onion routing provides security for the routing
process, it does not provide security for the route \emph{discovery} process.
In particular, by requesting shortest paths to certain destinations 
from the route discovery servers (e.g., the trampoline nodes in Lightning),
a node may reveal critical information about its planned financial transactions
\cite{blockchain20,icissp20}.

Motivated by this shortcoming of state-of-the-art payment channel networks,
this paper initiates the study of the design of scalable privacy-preserving 
route discovery mechanisms for payment channel networks.

The main contribution of this paper is the first scalable
and privacy-preserving route discovery mechanism 
for payment channel networks, \system\footnote{The name stands for private information retrieval for Lightning-like networks.}.
\system \space provides information theoretic security guarantees,
and is evaluated both analytically as well as empirically, using
the Lightning network.
We show that our approach is efficient and practical,
and may hence be useful also in other networks where the security of
route discovery is critical.
Concretely, we tested our approach with different snapshots of the Lightning network taken over a period of two years and our heuristic produces a small hub database (about 2MB), which at the moment is an improvement by an order of magnitude compared to the size of the complete network. With the growth of the network, this discrepancy will change further in favor of approaches similar to ours.

As a contribution to the research community and to ensure reproducibility, we will release the code of our implementation together with this paper.

The remainder of this paper is organized as follows. In Section~\ref{sec:probstatement} we formally define the problem we consider. Section~\ref{sec:approaches} presents several alternative solutions to our problem and highlights the potential drawbacks of each solution. We then proceed to state our main contributions in Section~\ref{sec:solution} which is an efficient solution that combines PIR and hub-labelling to solve our problem; our solution also relies on a heuristic which optimises the hub set sizes, to further improve efficiency. Section~\ref{sec:eval} contains an evaluation of our method on various snapshots of the Lightning network over a two year period. In Section~\ref{sec:relwork} we review related work, and finally we discuss future research directions in Section~\ref{sec:conclusion}.

\section{Problem Statement}\label{sec:probstatement}

We model a payment network by a public directed graph $G=(V,E)$, where each edge $e\in E$ has some cost $c(e)\ge 0$. The nodes are the clients and an edge $e=(u,v)$ means there's a channel available from $u$ to $v$ with transaction cost $c(e)$. 
We reserve the letters $n=|V|$ to denote the number of vertices, $m=|E|$ to denote the number of edges, $\delta=|E|/|V|$ to denote the average degree and $d$ for the edge length of the longest shortest path in $G$ (if all edge costs are identical this is just the diameter). We assume the nodes are labelled by $\lambda$-bit strings, while $\lambda=\lceil\log n\rceil$ is sufficient, we leave this as a parameter as one might want to use longer labels in practice.

Below we give our definition of a \emph{private route discovery mechanism}. The definition is very similar to the classic definition of private information retrieval (PIR): on the one hand, is more restricted as we only consider the particular problem of retrieving shortest paths; on the other hand it is more general as the client is not just supposed to recover some database entry, but can perform a more general computation on the retrieved answers to compute the path.

\subsection{Private Route Discovery Mechanism (PRDM) Definition}

The involved parties are a content provider $\CP$, a set of servers $\SV_1,\ldots,\SV_k$ and a client $\C$.

The content provider $\CP$ gets as input a graph $G=(V,E)$ and processes it into $t$ databases $\DB_1,\ldots,\DB_k$. $\DB_i$ is then sent to $\SV_i$.

The client $\C$ on input two nodes $s,t\in V$ can now interact with the servers and the {\bf correctness property} we require
that $\C$'s final output is a shortest path from $s$ to $t$ in $G$ (assuming $\CP$ and the $\SV_i$'s follow the protocol). 

The {\bf security guarantee} we require is $(\ell,k)$-privacy: the view of any of the $\ell>0$ out of the $k$ servers should not reveal any information about the query $s,t$.

We are interested in finding {\bf efficient} schemes, that is, the scheme has to scale well with the potential growth of the Lightning network.

\section{Possible Approaches and Drawbacks}\label{sec:approaches}

We start by giving a description of several schemes that fulfill both the correctness property and the security guarantee as described above but have varying degrees of efficiency. Note that for this problem efficiency can be defined in terms of several costs, namely, the amount of storage incurred by the server or the client, the amount of computation involved on the server or the client sides in order to get the final output, and the amount of data communicated between the client and the server. The amounts of storage incurred by the server and the client are most crucial parameters and our main focus in the following.  Another important cost is the amount of computation the client incurs in order to retrieve the desired shortest path, as client-side machines are typically not as powerful as server-side machines. We start with discussing the most naive solutions, gradually augmenting them with simple strategies to improve their efficiency. To provide a summary of the efficiency of each method, we present all the relevant asymptotic costs described above for each method in Table~\ref{tab:asymptotic}.

\begin{table*}
\small
\begin{center}
\begin{tabular}{ |c|c|c|c|c|c| } 
\hline
  & Server & Communi-- & Server & Client  & Client\\
 & storage & cation & computation  & computation & storage\\
\hline
Trivial solution & $\bigO(n\lambda \Delta)$ & $\bigO(n\lambda \Delta)$ & $\bigO(1)$ & $\bigO(m+n\log n)$ & $\bigO(n\lambda \Delta)$\\ 
\hline
APSP solution & $\bigO(n^2 d \lambda)$ & $\bigO(n^2d \lambda)$ & $\bigO(n^3)$ & $\bigO(1)$ & $\bigO(n^2 d \lambda)$\\ 
\hline
APSP + PIR & $\bigO(n^2 d \lambda)$ & $\bigO(n\sqrt{d \lambda})$ & $\bigO(n^3+n^2d\lambda)$ & $\bigO(n\sqrt{d\lambda})$ & $\bigO(1)$\\ 
\hline
APSP + Hub Labelling & $\bigO(nhd\lambda)$ & $\bigO(nhd\lambda)$ & $\bigO(n^3 + n^2 \max(d,$ & $\bigO(h)$ & $\bigO(nhd\lambda)$\\
&&&$\log n)\log d)$&&\\
\hline
\textbf{Our solution}& & & \bm{$\bigO(n^3+nhd \lambda $} & &\\
\textbf{(APSP + PIR + HL)}&  $\bm{\bigO(nhd \lambda)}$ & \bm{$\bigO(\sqrt{nhd\lambda})$}  & \bm{$+ n^2 \max(d,$} & \bm{$\bigO(h + \sqrt{nhd\lambda})$} & \bm{$\bigO(hd\lambda)$} \\
&&&\bm{$\log n)\log d)$}&&\\
\hline
\end{tabular}
\end{center}
\caption{Asymptotic server and client side storage, communication, and computation cost for the trivial solution, the APSP solution, the APSP solution with a PIR, the APSP solution with Hub Labelling, and our solution which is the APSP solution with a PIR and Hub Labelling (HL). $n$ is the number of vertices in the graph, $m$ is the number of edges, $\Delta$ is the maximum degree of each vertex, $\lambda$ is the length of the binary representation of each vertex, $d$ is the length of the longest shortest path in the graph, and $h$ is the maximal hub set size.}
\label{tab:asymptotic}
\end{table*}

\subsection{Trivial Solution: Downloading Entire Graph}

A trivial $(1,1)$-private solution is to simply have the client download the entire graph and compute the shortest path locally. Although the server does not need to perform any computation for this solution, the communication from server to client is significant, i.e., the entire graph of size $\approx \lambda\cdot n\cdot \delta$, and the client has to perform a non-trivial computation (e.g.,  Dijkstra which is $\bigO(m + n\log n)$) to retrieve the desired shortest path. In addition, the client also has to store the entire graph. Taking into account the fact that in most cases clients only sporadically send small amounts of queries, and also that the underlying network needs to be frequently updated to account for new users and channels, it is useful to push the heavy computation and storage burden to the server; the latter typically runs on machines with superior hardware and storage capacity.

\subsection{Server Computes All Pairs Shortest Paths}

A natural $(1,1)$-private solution which pushes the burden of computation onto the server is for the server to precompute and store all the all pair shortest paths (APSPs). This solution gives a database with $N=n(n-1)\approx n^2$ entries (every pair of nodes) with entries of length 
$L=\lambda\cdot d$, i.e., the $\lambda$ bit labels of the $\le d$ nodes for each path. 

Naturally the amount of computation the server would have to do to compute these APSPs ($\bigO(n^3)$) increases compared to the previous trivial $(1,1)$-private solution where the server does not have to perform any computation. Nevertheless this factor is not as crucial as the stored database size ($\bigO(n^2 d \lambda)$), which grows quadratically with $n$. Even for the moderate sized Lightning network with $n\approx 7000$, this gives an almost $2TB$ database (with $d=9$ and $\lambda=20$ bit labels). 

In addition, to maintain privacy, the client cannot simply query the server for the desired index, as this blatantly reveals the path information. The client would therefore have to download the entire database from the server to hide their query. Thus, the price of a constant client-side computation cost is an additional client-side storage and communication cost of $\bigO(n^2d\lambda)$ (see Table~\ref{tab:asymptotic}). 

\subsection{Trivial PRDM from Private Information Retrieval (PIR) Using APSPs}

In a $(\ell,k)$-private information retrieval (PIR) scheme we have $k$ servers, each holding a copy of a database $\DB$ which contains $N$ entries, each $L$ bits long. A client $\C$ who want to learn $\DB[i]$ for some $i$, computes queries $q_i$ to every $\SV_i$, and then can (efficiently) compute $\DB[i]$ from the answers $a_1,\ldots,a_k$. The security property ($\ell$-privacy) requires that any union of $\ell$ servers cannot learn anything about the query $i$.

The paper introducing PIR~\cite[Corollary 4.2.1.]{ChorKGS98} has a particularly simple $(1,2)$-private PIR where the communication complexity is 
just $4\max\{L,\sqrt{N\cdot L}\}$, in particular, that's $4L$ if $L\ge N$.\footnote{In this construction we think of $\DB$ as an $L\times N$ matrix when $L\ge N$, $q_1\in\{0,1\}^N$ is random and $q_2=q_1\oplus e_i$ where $e_i$ is the zero vector with the $i$th position set to $1$. From the answers $a_i=\DB\cdot q_i$ the client computes 
$a_1\oplus a_2=\DB\cdot e_i=\DB[i]\ni \{0,1\}^L$. 
 If $L<N$ one does almost the same but thinks of $\DB$ as a 
$\sqrt{N\cdot L}\times \sqrt{N\cdot L}$ matrix.}

We can get a $(\ell,k)$-private mechanism using a $(\ell,k)$-private information retrieval on top of the APSP solution. Using a PIR protocol similar to~\cite{ChorKGS98} and assuming that the storage size of each shortest path is smaller that the total number of entries in the database, i.e. $d\cdot \lambda < n^2$, this solution reduces the communication cost between the client and servers from $\bigO(n^2d\lambda)$ to $\bigO(n\sqrt{d\lambda})$ as shown in Table~\ref{tab:asymptotic}. In addition, the client now does not need to download and store the database but just has to perform $\bigO(n\sqrt{d\lambda})$ XOR operations to get the desired shortest path so the client-side storage cost is a constant. 

From Table~\ref{tab:asymptotic}, we see that this solution incurs a larger asymptotic computation cost on both the server and client side compared to the vanilla APSP solution. Specifically, the server has to perform $\bigO(n^2 d\lambda)$ XOR operations upon each query from the client in addition to the initial $\bigO(n^3)$ computations to get the APSPs. These cost increases are not too alarming as, firstly, server-side computation costs are not as important as user-side computation and storage costs. More importantly, the length of the longest shortest path in the Lightning network $d$ is very small and unlikely to increase much as the network grows due to the centralised nature of the network (see Section~\ref{sec:topology}), and we can reasonably assume $\lambda \ll n$. 

The server-side storage requirements of this solution is unfortunately the same as the vanilla APSP solution and is the main drawback of this approach. In addition, the PIR protocol requires the servers to perform XOR computations linear in the size of the entire database in response to every query. Thus it would be highly  beneficial if the database could be kept in memory. As outlined above, for Lightning that is already not possible with this approach.

\subsection{Hub Labelling}

In a hub labelling scheme~\cite{abraham2011hub} one preprocesses a graph such that later, one can efficiently find shortest paths in-between any two nodes. For some graphs, in particular street networks, the preprocessed data is \emph{much} smaller than storing all pairwise shortest paths.

Concretely, given $G=(V,E)$ the idea is to precompute APSPs. Then compute for every vertex $v\in V$ a set of hubs $hub(v)\subset V$ such that for all $u,v\in V$, the shortest path from $u$ to $v$ (if $G$ is directed also from $v$ to $u$) contains a vertex that lies in the intersection of their hub sets $hub(u)\cap hub(v)$. This is the \emph{covering property} of the hub sets.

The preprocessed $\DB$ now contains, for every $v\in V$ its hub set $hub(v)$ and for every $u\in hub(v)$ the shortest path from $v$ to $u$ (and $u$ to $v$).

If $h=\max_{v\in V}\{|hub(v)|\}$ is the size of the largest hub, we can think of the preprocessed data $\DB$ as a database with $n=|V|$ entries of length $L=h\cdot\lambda \cdot d$ bits, i.e., for every $v\in V$ we have $hub(v)\le h$ paths, each of length $\le d$, and thus the asymptotic server-side storage drops from $\bigO(n^2d \lambda)$ to $\bigO(nhd\lambda)$. For $h \ll n$, which is the case for the Lightning network (see Section~\ref{sec:eval}), this allows us to achieve at least an order of magnitude reduction in server-side storage costs.

Compared to the vanilla APSP solution, adding hub-labelling increases the asymptotic server-side computation by an additive factor of $n^2\max(d, \log n)\log d$. The increase in computation cost is due to the hitting set computation for each radius. Nevertheless we note that the increase is not big as $d \ll n$ in the Lightning network and so the main cost is still the $\bigO(n^3)$ from the APSP computation. In addition, server-side computation is not as important as server-side storage, and these computations are also not costs incurred with every query but need only to be be done relatively infrequently, for instance every two weeks or so to take into account changes in the underlying network topology.

The main drawback of this solution is that the client needs to download the entire database to maintain the security guarantee. This would incur an asymptotic storage cost of $\bigO(nhd\lambda)$ which is problematic as client-side machines have less storage capacity compared to server-side machines.

\section{The LightPIR Approach}\label{sec:solution}

\begin{figure}[t]
    \centering
    \includegraphics[scale=0.7]{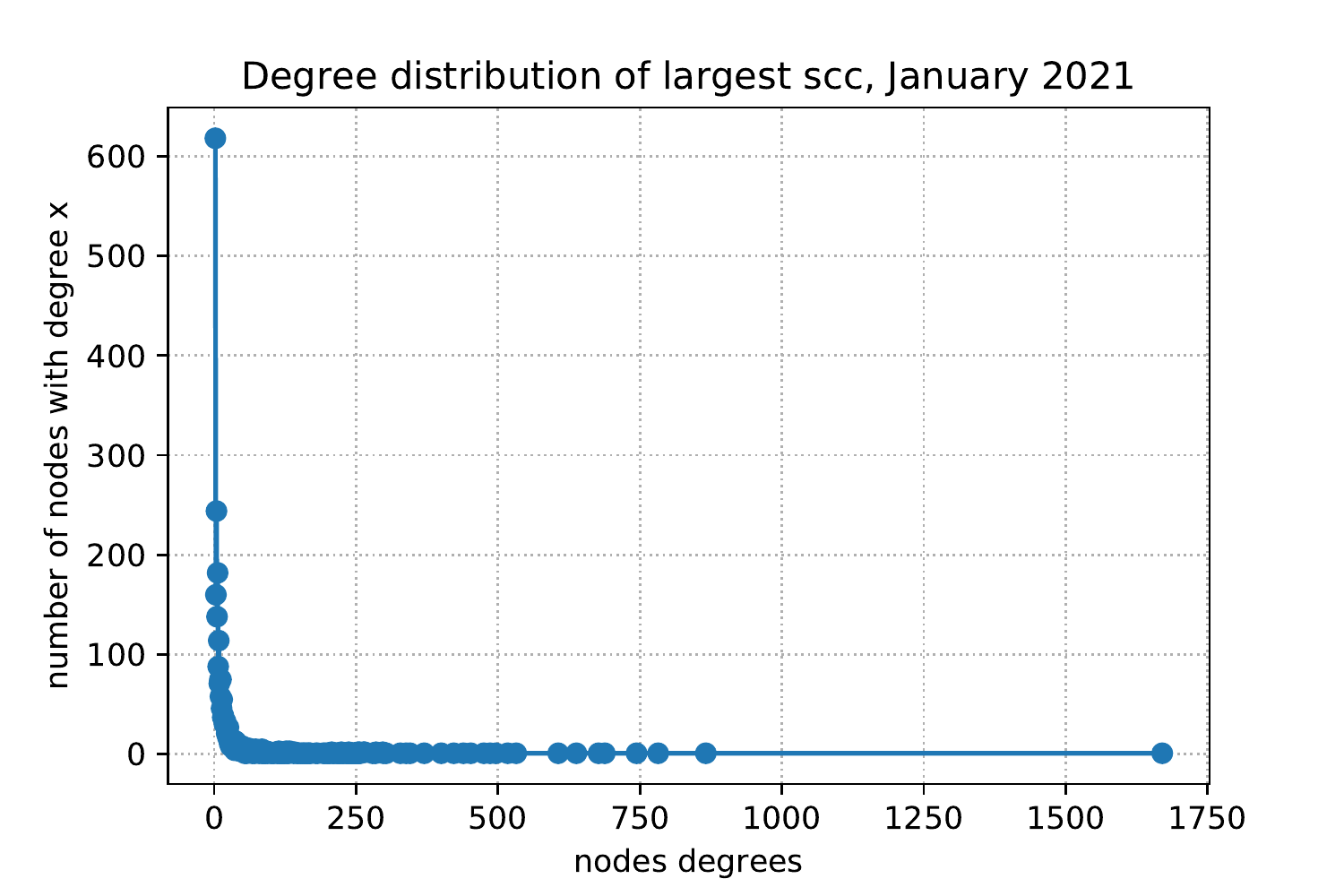}
    \caption{A plot showing the number of nodes for each degree in the largest strongly component (scc) of the Lightning network on January 2021. The distribution of degrees is quite similar in Lightning network snapshots taken in dates within the past two years.}
    \label{fig:degdist}
\end{figure}

Given the above considerations, we are now ready to
present our proposed solution, \system.

\subsection{Basic Approach}

In a nutshell, our construction of a private route discovery mechanism augments the basic APSP solution with both a PIR protocol (specifically a 2-server RAID-PIR protocol~\cite{DemmlerH2014Raidpir} which is similar in flavour to~\cite{ChorKGS98}) and hub labelling. A client who wants to learn the shortest path from $u$ to $v$ will query the PIR servers for the $u,v$ entries, which will be the hub sets of $u$ and $v$ respectively. Then, the client should perform a set intersection of the two hub sets, which will be non-empty by the covering property of these hubs, and determine the $w\in hub(u)\cap hub(v)$ for which the length of the path $u\rightarrow w$ plus 
the length of $w\rightarrow v$ is minimized. The path 
$u\rightarrow w\rightarrow v$ is the sought shortest path from $u$ to $v$.

\system \space reaps the benefit of a lower asymptotic storage cost of $\bigO(nhd\lambda)$ from hub labelling and a lower asymptotic communication cost of $\bigO(\sqrt{nhd\lambda})$ from the 2-server PIR protocol. The composition of these two techniques also lowers the per query asymptotic server-side computation due to the PIR protocol from $\bigO(n^2d\lambda)$ XOR operations to $\bigO(nhd\lambda)$. Compared to the vanilla APSP solution, the user now has to perform some computation to determine the desired shortest path. This computation is extremely simple, basically $\bigO(\sqrt{nhd\lambda})$ XOR operations to privately recover the query result from the servers and a set intersection on two moderately sized sets of at most $h \ll n$ to get the shortest path, and thus can be performed by even the weakest clients. Finally we note that although the asymptotic server and user side computation costs is higher in our solution compared to the APSP solution with hub labelling, the asymptotic storage incurred by the user drops from $\bigO(nhd\lambda)$ to $\bigO(hd\lambda)$ since the user just has to store the hub sets of the desired source and target nodes instead of all the hub sets of all nodes. Again, since $h \ll n$ in the case of the Lightning network, this is at least an order of magnitude drop in client-side storage requirements.

\subsection{Optimization of Hub Sets}

In order to lower the database storage further, \system\ relies on a heuristic which leverages the specific topology of the Lightning network to optimize the size of hub sets. We first begin with a brief description of the Lightning network topology and then describe the heuristic.

\subsubsection{Network Topology Characteristics}\label{sec:topology}
To understand and exploit the specific topological characteristics
of payment channel networks, we downloaded a number of historical snapshots of the Lightning network from the Lightning github repositories~\cite{lightninggitrepo, lngossip}. For instance, a representative network snapshot on January 2021 comprises of 6,376 nodes, 39,993 channels, and 3,650 strongly connected components, although except for the largest one, almost all of the rest components have only a very small number of nodes. We thus focus our analysis on the largest connected component of each snapshot. The one of January 2021 has 2,707 nodes and comprises a majority (31,350) of the total number of channels.

As can be seen from Figure~\ref{fig:degdist}, the majority of the nodes in the largest component of the 
Lightning network have very low degrees and about 100 nodes have degrees of $>100$. The largest component also has a relatively small diameter of 9. We can therefore infer a large degree of centralisation in the Lightning network, with a few central nodes of high degree connected to each other as well as connected to shorter chains of vertices (see Figure~\ref{fig:Hub2}). 

A channel (i.e. an edge) between two nodes in the network is weighted by its channel fee. This usually comprises of a base fee which is independent of the amount one wishes to move through the channel, and an amount-specific rate fee which depends on the amount sent across the channel. 
Table~\ref{tab:all-results} (Section \ref{sec:eval}) includes a number of network characteristics for a number of snapshots of the Lightning network over the past two years.



\begin{figure}%
    \centering
    \subfloat[\centering Figure \ref{fig:Hub2}]{{\includegraphics[width=0.45\linewidth]{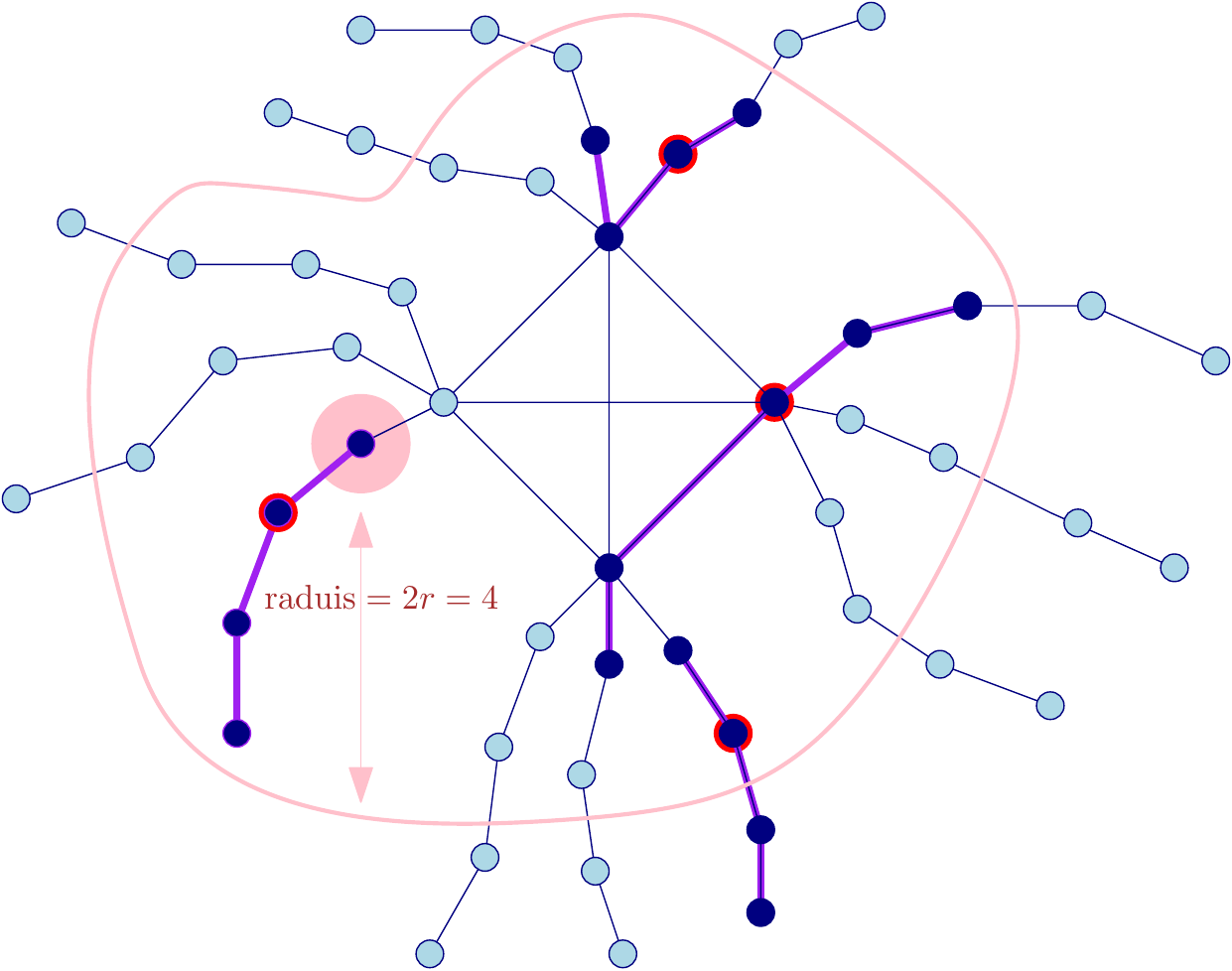}\label{fig:Hub2} }}%
    \qquad
    \subfloat[\centering Figure \ref{fig:Hub1}]{{\includegraphics[width=0.45\linewidth]{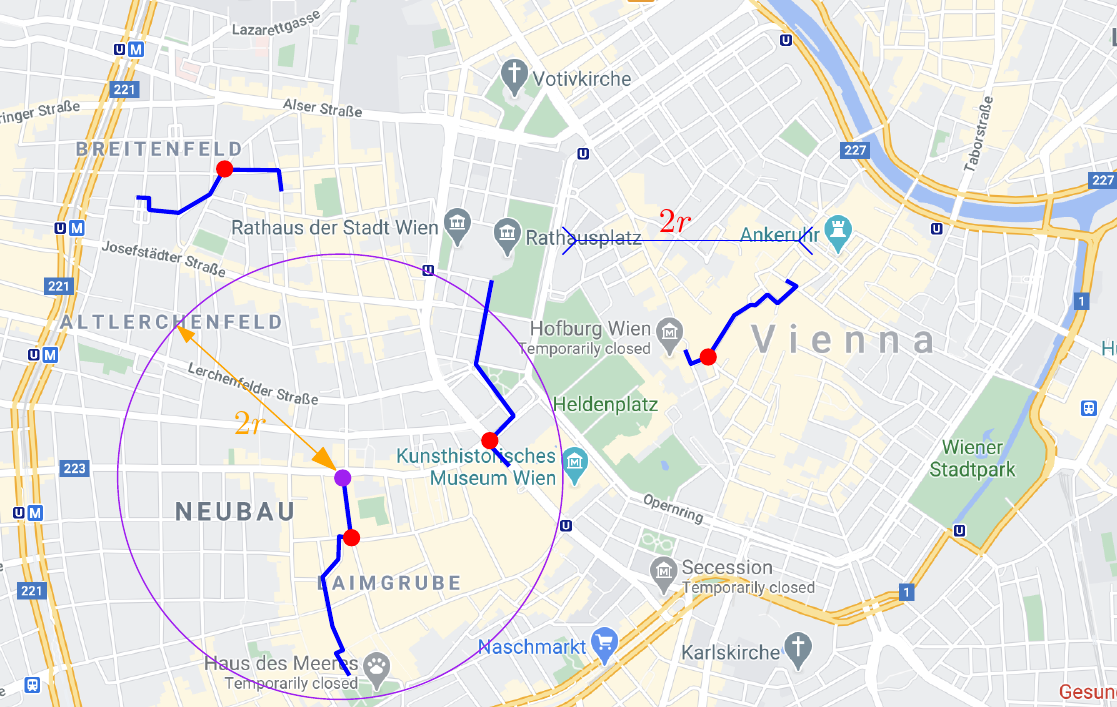}\label{fig:Hub1} }}%
    \caption{Figure \ref{fig:Hub2} shows a graph created from star graphs where the centers of the graphs form a clique. A few paths of length in $[r+1,2r]$ for $r=2$ with a (red) cover point on each path are shown. A ball (neighborhood) of radius $2r=4$ is shown, and it contains a substantial fraction of the graph, and thus the cover points.
    Figure \ref{fig:Hub1} shows a street network with 4 shortest paths of length in $(r,2r]$ (for $r\approx 500m$) and a cover point of those paths (in red) highlighted. A ball (neighborhood) with radius 2r (one with center at the starting point of a path shown in purple) will only contain cover points resulting from relatively few paths ``in the neighbourhood".
    }%
    \label{fig:hubfigs}%
\end{figure}

\subsubsection{Hub Set Computation: Greedy Heuristic}

\LinesNumbered
\begin{algorithm}[t!]
\small
\KwInput{$G = (V,E)$, the network topology and $\ell$, the number of high degree nodes to use as base}

\KwOutput{$\HDbound(G)$ (bound on $\HD(G)$) and $hubs(u),\, \forall u\in V$}

$APSP \gets allPairsShortestPaths(G)$\;

$radii \gets \{2^i \,|\, i \in \N \land 2^i \leq diameter(G) \}$\; 

$base \gets highDegreeNodes[0:\ell-1]$\;

\tcc{compute shortest path sets}
\myFor{$r \in radii$}{$shortestPaths(r) \gets \emptyset$}\tcc*{init}
\For{$path \in APSP$}{
\For{$r \in radii$}{
\If{$r<|path|\leq 2r \land base \cap path = \emptyset$}{
$shortestPaths(r).add(path)$\;
}
}
}

\tcc{compute hitting sets and covers}
\For{$r \in radii$}{$cover[r] \gets base \cup hittingSet(shortestPaths(r))$\;}

\tcc{compute $N_{2r}^{\overline{G}}(u)\cap cover[r]$ sizes \& hubs}
$\overline{G} \gets G \setminus base$\;
\myFor{$u\in V$}{$hubs(u) \gets \emptyset$}\tcc*{init $hubs$}
\For{$(u,r): u\in V \land r \in radii$}{
$hubs\_in\_neighb(u,r) \gets base \cup (N_{2r}^{\overline{G}}(u) \cap cover[r])$\label{algln:hubsinN}\;
$hubs(u) \gets hubs(u) \cup (hubs\_in\_neighb(u,r))$\;
}

$\HDbound(G) = \max_{u,r} |hubs\_in\_neighb(u,r)|$\;
\caption{\system's hub set optimization}
\label{alg:HD}
\end{algorithm}


One way of computing the hubs for each node in the network is to make use of the highway dimension (HD) algorithm~\cite{abraham2010highway, abraham2011vc}. 
Intuitively, HD is a graph metric that captures the size of a subset of nodes that intersects all shortest paths of non-trivial length. 
In order to compute the HD, we compute sets of nodes that intersect all paths of size $(r,2r]$ for a small set of radii and measure how sparse these sets are in the network. 
In the following, we explain how we use these sets to compute the hubs for each network node (Algorithm~\ref{alg:HD}).

In~\cite{abraham2010highway}, the authors give a definition of HD based on shortest-path covers and also give an approximation algorithm to compute the HD of a network using this definition. In~\cite{abraham2011vc}, the authors present a modification that improves their earlier work, which also applies to the shortest-path covers algorithm. We further modify this approximation algorithm with a heuristic from observing the topology of the Lightning network, as shown in Algorithm~\ref{alg:HD}, to both return a hub set for each node in the network, and also maintain an upper bound on the HD of the network. That is, we always overapproximate the HD but never underapproximate it. 

We present a few necessary ingredients from~\cite[Section 4]{abraham2011vc}, before elaborating on our heuristic. Let $G = (V,E)$ be the network topology.
The neighborhood $N_x(u)$ of a node $u \in V$ for a radius $x$ is the set of all nodes that appear in a path $P$, such that $v\in P$ and $|P| \leq x$, where $|P|$ denotes the length of the path $P$ ($N_x(v)$ contains the same number of nodes as the ball with center $u$ and radius $x$).
Moreover, a set $C$ of nodes is a $(k,r)$-shortest path cover ($(k,r)$-SPC) of $G$ if and only if (I) for each shortest path $P$, such that $r < |P| \leq 2r$, $P \cap C \neq \emptyset$ and (II) $\forall u\in V$, $|C \cap N_{2r}(u)| \leq k$. 
Then, if $G$ has HD $h$, then for any $r>0$ there exists an $(h,r)$-SPC of $G$.
Thus, the authors claim that by using the greedy $\bigO(\log n)$ approximation for computing hitting sets, we can combine the above into an algorithm that computes a bound of HD in $\bigO(h \log n)$, where $h = \HD(G)$ is the highway dimension of $G$.
Specifically, given a hitting set $cover[r]$ of all pairs shortest paths of length in $(r,2r]$, they make sure that condition I holds and then they can fulfill Condition II by taking all $cover[r] \cap N_{2r}(v)$ intersections for all $v \in V$ and setting $k$ to be the maximum size of those intersections.
As mentioned in~\cite{abraham2011hub}, it suffices to compute the above with the set of radii $r \in \{2^i \,|\, i \in \N \land 2^i \leq diameter(G) \}$. 
Of course, an exact hitting set algorithm would give a smaller $cover[r]$, but computing it is NP-hard.

Our heuristic is based on the key observation that, due to the centralised nature of the network, high degree nodes in the Lightning network appear in a vast majority of shortest paths, as opposed to the less dense and centralised road networks for which HD was originally designed (Figure~\ref{fig:Hub1}).
Thus many neighborhoods for center $u$ and radius $2r$ include these high degree nodes, which work as gateways to many remote nodes.
We denote by $base(\ell)$ the set of the $\ell$ nodes with the highest degree (in short, $base$).
We noticed that for network topologies similar to the one of Lightning removing the base nodes after computing the covers can significantly reduce the HD (Figure~\ref{fig:Hub2}).
Specifically, we adapt the HD computation by first computing the cover sets as in~\cite{abraham2010highway, abraham2011vc}, but then removing the base nodes and all the edges connected to the base nodes to obtain the topology $\overline{G} = G \setminus base$, in which we compute the per-node hubs as follows.
That is the hub set of node $u$ for radius $2r$ is defined by $hubs\_in\_neighb(u,r) = base \cup (N^{\overline{G}}_{2r}(u) \cap cover[r])$, where $N^{\overline{G}}_{2r}(u)$ is the neighborhood $N_{2r}(u)$ computed in $\overline{G}$.
Note that~\cite{abraham2011vc} (in addition to~\cite{abraham2010highway}) requires that the path weights are unique integers, but this is easily achieved with insignificant perturbations of edge weights (in our case we initially set the weights to be the channel base fees and then perturbed them).
We illustrate this heuristic in Algorithm~\ref{alg:HD} and show that it correctly computes an upper bound $\HDbound(G)$ of $\HD(G)$ in Lemma~\ref{lem:heur}.

\begin{lemma}
\label{lem:heur}
Algorithm~\ref{alg:HD} computes an upper bound of $\HD(G)$.
\end{lemma}

\begin{proof}
We prove that for any two communicating nodes $s$ and $d$, there exists a hub node in the set of hubs that lies on the shortest path between $s$ and $d$, $SP(s,d)$.
Let $|SP(s,d)| = k$ and let $r$ be the radius for which $r < k \leq 2r$ (there has to be at least one such radius as our set of radii covers all shortest paths).
If $SP(s,d) \subset N_{2r}^{\overline{G}}(s)$, i.e. if $SP(s,d)$ entirely appears in the neighborhood with center $s$ and radius $2r$ in the resulting topology \textit{after} removing $base$ from $G$, then the hub for $s$ and $d$ is in $SP(s,d) \cap cover[r]$, as $cover[r]$ hits all shortest paths of length in $(r,2r]$.
Otherwise, if $SP(s,d)$ does not entirely appear in $N_{2r}^{\overline{G}}(s)$, then there exist at least one node in $base$ that intersects $SP(s,d)$. Since $base$ is included in the set of hubs (Algorithm~\ref{alg:HD}, line~\ref{algln:hubsinN}), then the proof is complete.
\end{proof}

Note that from the proof of Lemma~\ref{lem:heur}, we guarantee the covering property and hence correctness of the hub sets: the hub set of any node covers all shortest paths from that node to any other node in the network. This means that any two hub sets in a connected network must certainly have a non-empty intersection.

The time complexity of Algorithm~\ref{alg:HD} is computed as follows. The APSP step takes $\bigO(n^3)$ and the partitioning of APSP according to their length takes $\bigO(n^2\log d)$.  The computation of covers, which includes the computation of hitting sets (greedy heuristic costs $\bigO(n^2 \max(d,\log n))$), takes $\bigO((n^2 \max(d,\log n))\cdot\log d)$ and the hub computation takes $\bigO(n\log d)$.
Thus, in total $\bigO(n^3 + n^2\log d + (n^2 \max(d,\log n))\log d + n\log d) = \bigO(n^3 + n^2 \max(d,\log n)\log d)$ (cf. Table~\ref{tab:asymptotic}). 
In practice, the diameters of the Lightning network snapshots that we used to test Algorithm~\ref{alg:HD} were quite small, thus the running time was dominated by the APSP computation (which was a matter of minutes).

\section{Empirical Evaluation}\label{sec:eval}

\begin{table*}
\footnotesize
\centering
\caption{Lightning Network Snapshot Characteristics and Evaluation Results (scc stands for strongly connected component(s), $V$ for the set of nodes, $E$ for the set of channels, and $D$ for the diameter)}
\begin{tabular}{|c|c|c|c|c|c|c||c|c|c|c|c|c|}
\hline
 & & & & \multicolumn{9}{|c|}{{largest scc}} \\ \cline{5-13}
{date} & {$|V|$} & {$|E|$} & {scc} & & & & {baseline} & {heuristic} & {base} & {max}  & {hub DB} & {exec.} \\
&  &  &  & {$|V|$} &  {$|E|$} &  {$D$} & {HD} & {HD} &  {size} & {hub set} &  {size} & {time} \\ 
&  &  &  & &  &  & {bound} & {bound} &  & size &  (MB) & {(h)} \\ \hline
13-03-2019 & 3480 & 45071 & 994 & 2480 & 42025 & 7 & 386 & 209 & 100 & 271 & 1.5 & 2.03\\ \hline
13-07-2019 & 4469 & 53084 & 1804 & 2661 & 46733 & 7 & 312 & 191 & 100 & 248 & 1.4 & 1.41\\ \hline
13-09-2019 & 4733 & 51508 & 2027 & 2706 & 45774 & 7 & 389 & 228 & 100 & 304 & 2.1 & 2.64\\ \hline
13-11-2019 & 4598 & 44784 & 2158 & 2422 & 38011 & 6 & 306 & 186 & 110 & 235 & 1.2 & 1.31\\ \hline
13-03-2020 & 5070 & 48340 & 2356 & 2712 & 41571 & 6 & 368 & 212 & 100 & 270 & 1.7 & 2.67\\ \hline
13-05-2020 & 5557 & 49569 & 2570 & 2985 & 42093 & 8 & 347 & 196 & 100 & 257 & 1.6 & 1.98\\ \hline
13-07-2020 & 5888 & 51932 & 2748 & 3136 & 44005 & 7 & 361 & 211 & 100 & 286 & 1.8 & 2.27\\ \hline
13-09-2020 & 5972 & 51783 & 2800 & 3171 & 43844 & 7 & 373 & 220 & 112 & 285 & 1.8 & 2.47\\ \hline
13-11-2020 & 6074 & 48361 & 2978 & 3089 & 40695 & 8 & 419 & 253 & 103 & 315 & 2.3 & 2.96 \\ \hline
13-01-2021 & 6376 & 39993 & 3650 & 2707 & 31350 & 9 & 422 & 259 & 106 & 314 & 1.8 & 1.83\\
\hline
\end{tabular}

\label{tab:all-results}

\end{table*}

To provide a first empirical evaluation of our approach, we implemented \system\ (the code will be released together with this paper) and ran experiments on historical snapshots of the Lightning network from the Lightning network gossip repository~\cite{lngossip}. 
In general, we find that our approach indeed achieves an improvement by an order of magnitude in terms of storage requirements and a roughly 40\% reduction in the highway dimension of each snapshot, which directly translates to a reduced hubs database. 

\paragraph{Setup}
In our experiments, we used a virtual machine with 24 cores (Intel Xeon CPUs E5-2650 v4 at 2.20GHz) and
16 GB RAM, running Ubuntu 18.04.1 LTS. Our experiments were implemented in Python (version 3.7.6) using networkx~\cite{SciPyProceedings_11} among other libraries.

\paragraph{Data} The extracted snapshots range from March 2019 to the current state (January 2021) with the smallest snapshot having 3,480 nodes and the largest having 6,376 nodes. We ignore earlier snapshots as they have too few nodes. We performed our empirical evaluation on the largest strongly connected component (SCC) in the snapshot and we found that although the network size almost doubles from March 2019 to January 2021, the number of nodes in the largest SCC remains pretty consistent across all snapshots ($\approx 2500$ nodes).
The first seven columns of Table~\ref{tab:all-results} give a complete view of the datasets used.

\begin{figure}[t]
    \centering
    \includegraphics[scale=0.7]{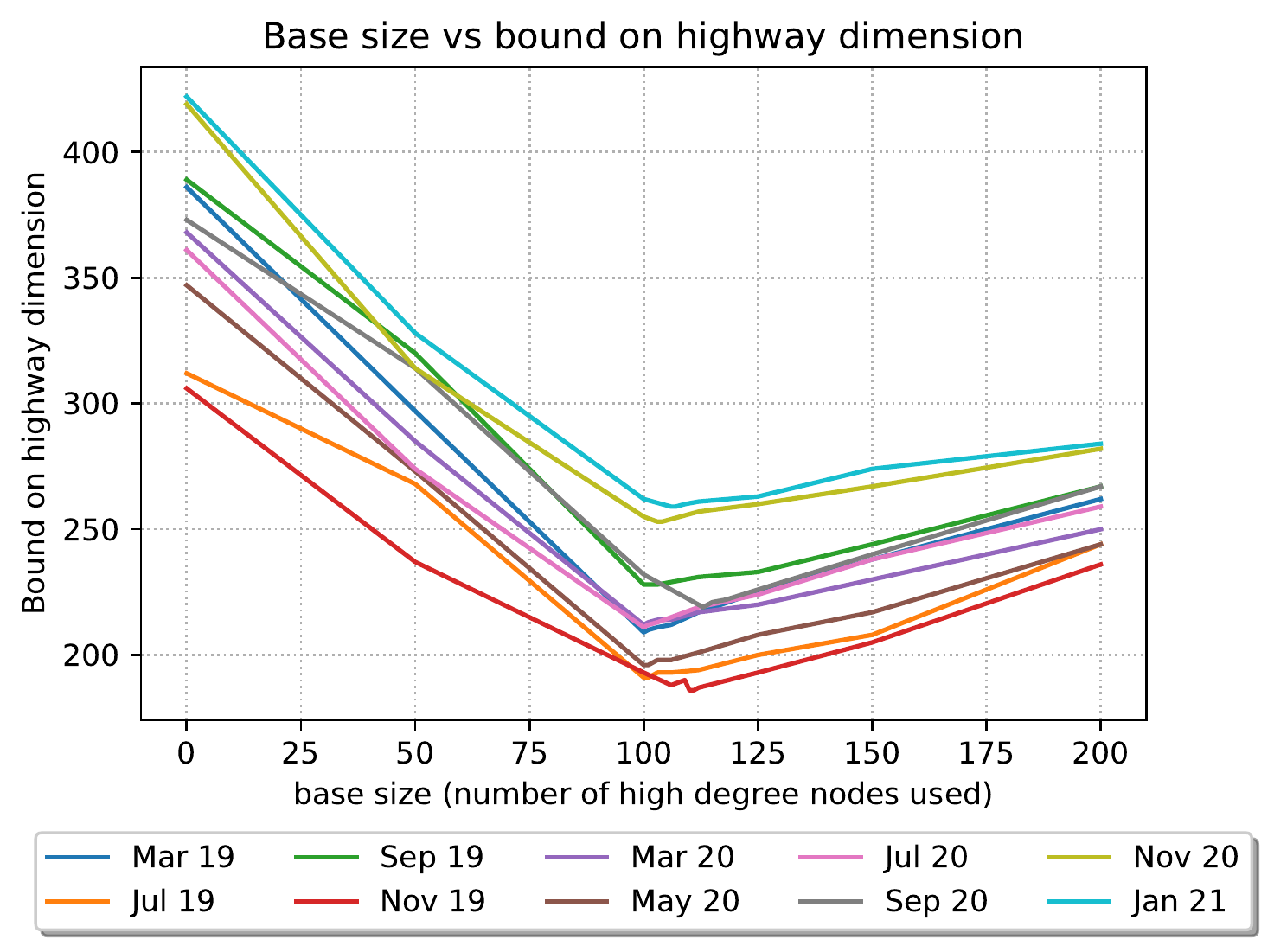}
    \caption{The x-axis shows the different base set sizes tested (a base with size $\ell$ includes the $\ell$ nodes with the highest degree) and the y-axis shows the corresponding highway dimension bound. The baseline computation corresponds to the values for base size 0, which are the highest. We run our heuristic with the base size that minimizes the highway dimension bound (thus also the hub database) and reported the results in Table~\ref{tab:all-results}.}
    \label{fig:HDvsBase}
\end{figure}

\begin{figure}[h]
    \centering
    \includegraphics[scale=0.7]{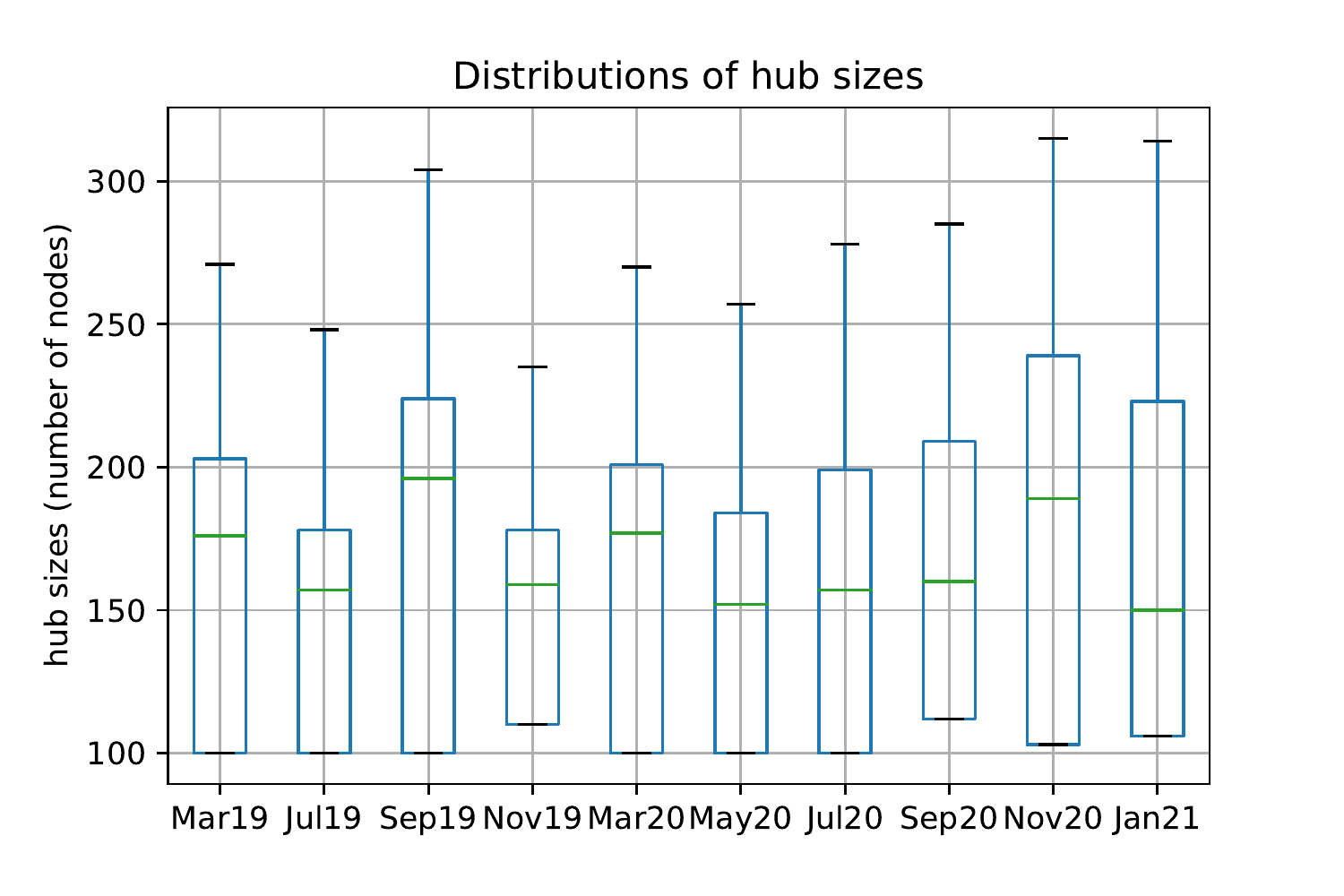}
    \caption{Box plots of hub sizes (in number of nodes) for a number of snapshots of the Lightning network. The green line shows the second quartile (median), the bottom and top of each box show the first and third quartiles, respectively, and the extended lines out of each box (whiskers) end at the min (bottom) and max (top) values of the datasets.}
    \label{fig:HD}
\end{figure}

\paragraph{Results}
We computed the optimal base size for the largest SCC in each snapshot by running a binary search to find the size of the base set that minimises the approximate highway dimension of the SCC as computed by our algorithm ($\ell$ in Algorithm~\ref{alg:HD}'s input). 
The baseline algorithm of~\cite{abraham2010highway,abraham2011vc} is the one where $\ell = 0$, i.e., there the base set is empty.
As shown in Figure~\ref{fig:HDvsBase}, for most Lightning network snapshots a base size of around $100$ nodes minimises the approximate highway dimension of the largest SCC. In fact, the reduction in highway dimension is around 40$\%$ across all snapshots when compared to the baseline (see Table~\ref{tab:all-results} baseline HD bound vs heuristic HD bound).

Across all snapshots the median hub set size as computed by our algorithm is between 150 to 200 nodes, with the maximum hub set size being 315 nodes. The distribution of hub sizes for each snapshot is shown in Figure~\ref{fig:HD}. Given that the average size for storing each snapshot is about 13MB, our method provides an overall reduction of an order of magnitude in terms of the hub set sizes (cf. second last column of Table~\ref{tab:all-results}). Due to the inclusion of the base, the minimum hub set size is at least 100. 
Specifically, as shown in Table~\ref{tab:all-results}, the average database size across all snapshots is 1.65MB.

The running time for each snapshot is about two hours (last column of Table~\ref{tab:all-results}).
This includes the search for the base size that minimizes the HD (and thus the hubs database).
Thus, it is possible to run our algorithm periodically to update the hub sets as the lightning network topology changes.

\section{Related Work and Novelty}\label{sec:relwork}

While the security of blockchains, e.g., related to the underlying consensus 
protocols~\cite{gervais2016security,wust2016ethereum}  
but also related to the interconnecting network~\cite{gervais2015tampering}, has been extensively investigated,
much less is known about the security of the recent 
concept of payment channel networks.
The security aspects of payment channel networks considered so far
mainly revolve around connectivity~\cite{rohrer2019discharged}
and the multi-hop routing of payments~\cite{aft20,malavolta2017concurrency,icissp20}.
We refer the reader to the surveys~\cite{gudgeon2020sok,csur21crypto} for an overview of specific 
threats related to hot wallets, 
mass exits, synchronizing protocols, among others.

Prior work on route discovery primarily focused on 
efficiency. For example, SpeedyMurmurs~\cite{roos2017settling}
(which extends VOUTE~\cite{roos2016voute})
and SilentWhispers~\cite{malavolta2017silentwhispers} rely on clever protocols
to speed up route discovery by
two orders of magnitude,
while maintaining the same success ratio. 
Another interesting approach is
pursued in the SpiderNetwork~\cite{sivaraman2018routing}: the payments are split
into units and the route discovery algorithm routes each of
them individually (similarly to a packet-switched network).
This method however also does not account for privacy and cost effectiveness. 
To trade off high throughput and probing overhead 
Flash~\cite{wang2019flash} distinguishes between mice and elephant payments,
splitting elephant flows across multiple flows. 
Flare~\cite{prihodko2016flare} proposes to improve scalability of route discovery
by determining beacon nodes, an approach which also motivates our model in this paper.
To overcome the cost overheads of multi-hop routing, 
thee off-chain channel network  Perun~\cite{dziembowski2019perun} 
introduces a notion of virtual payment channels;
while this  allows to avoid intermediaries for each payment, it does not solve
the discovery problem.  

Private route discovery on street networks has been explored in~\cite{Wu2016}. Their work is similar to ours in that they also use a PIR protocol to privately retrieve shortest path entries in a database. The main difference is the database compression technique: in~\cite{Wu2016} the authors develop a compression algorithm specific to the grid-like structure of street networks in North America (i.e. each node has degree 4 corresponding to the 4 cardinal directions), whereas we use hub labelling which has good compression guarantees on generic networks and we develop an additional heuristic based on the topology of the Lightning network to get small hubs set sizes. 

Our work is motivated by recent empirical work demonstrating issues with confidentiality 
during the route discovery process in the Lightning network~\cite{blockchain20,icissp20}.
To the best of our knowledge, the only work suggesting a more secure route discovery
mechanism is MAPPCN~\cite{tripathy2020mappcn},
which however does not account of the
information leaking problem addressed in this paper.

Our paper also builds upon classic works on efficient shortest path computations
on extremely large networks~\cite{Goldberg2005astar, Lauther2006path, Gutman2004reach}. Of particular interest in our context are algorithms based on transit node routing (TNR)~\cite{Bast2007tnr}, which focus on networks in which a small set of vertices covers most shortest paths, i.e., networks with small so-called highway dimension. Recently, Abraham et al.~\cite{abraham2010highway, abraham2011vc} proved theoretical guarantees on the efficiency of common shortest path search algorithms for graphs which satisfy this property. The hub-labelling algorithm~\cite{abraham2011hub, Abraham2012hhl} is a variant of TNR and it was shown to be significantly faster compared to other state of the art shortest path algorithms when tested on real world road networks~\cite{abraham2011hub}.
In this paper, we build upon this approach to scalable routing, tailoring it to the specific  topological properties of payment channel networks.

Information theoretic PIR protocols were first introduced by Chor et al.~\cite{ChorKGS98}. With 2 servers and a simple linear bit summation scheme, they achieve perfect privacy with a communication complexity of $\bigO(n^{1/3})$. Since then, follow up work in this area has mainly focused on lowering the communication complexity~\cite{Beimel2001pir, Beimel2002pirbarrier, dvir20162, Efremenko2012pir, DemmlerH2014Raidpir}. Recently, Beimel et al.~\cite{Beimel2002BreakingTO} achieved a communication complexity of $n^{\bigO(\frac{\log \log k}{k \log k})}$ with $k$ servers using locally decodable codes. Under the conjecture that there are infinitely many Mersenne primes, Yekhanin~\cite{Yekhanin2008pir} removed the dependency on large $k$ in the communication complexity bound and achieved a communication complexity of $n^{\bigO(\frac{1}{\log \log n})}$ for a 3 server PIR protocol. Although these newer protocols achieve a much lower communication complexity compared to the original protocol, the techniques used in these protocols are also much more complicated. As these benefits from a lower communication complexity are not that pertinent to us, we focus on simpler PIR protocols which are similar to~\cite{ChorKGS98} in this paper.

 PIR protocols
have already
been successfully deployed in many contexts, such as DNS~\cite{Zhao2007}, private presence service~\cite{Borisov2015DP5AP}, and private retrieval of security updates~\cite{Cappos2013secupdates}. 
However, we are not aware of any applications of PIRs 
in the context of payment channel networks.

\section{Future Work}\label{sec:conclusion}

We understand our work as a first step and would like to highlight three interesting directions of future work. 

Firstly, when shortest paths are not unique, most available algorithms for selecting a shortest path out of a set of paths of equal length just select a random one from the set. In terms of reducing the size of the hub set, one can design smarter algorithms that do not select the path randomly, but optimise for paths with a larger vertex overlap with the set of other paths already selected. This ensures that when we greedily select vertices to be in the hub set, we would select vertices that cover a larger set of paths and so could potentially decrease the size of the hub set. 

Secondly, our current heuristics to calculate the highway dimension of the network and to reduce the size of the hub sets are heavily inspired by the topology of the Lightning network. It would be interesting to track the topology of the Lightning network as the network grows in size, to see if there are significant changes and to think about how to incorporate these changes in the network topology to further optimise the hub set sizes. 

Lastly, we note that the network topology and thus the hub sets change with the amount of Bitcoins a user wants to send, e.g., due to the  capacity limitation of channels in the network (i.e. the amount of Bitcoins each channel can forward). For instance the shortest path from Alice to Bob might go through Charlie when Alice wants to send 1 Btc but the path might have to go through Dave when Alice wants to send 10 Btc because Charlie cannot forward more than 1 Btc. Additionally, the topology could also change due to the way the fees are calculated. The fees on channels comprise of a base fee and a rate fee which depends on the amount sent over the channel and thus the shortest path information could change drastically between sending small or large amounts of Bitcoin. Since we do not yet take into account fees, our current work can be seen as optimised for sending a fixed amount of Bitcoin and thus an interesting future direction would be to take into account the amount of Bitcoins a user wants to send and design a similarly efficient and private method of extracting shortest paths.

{\balance
\bibliographystyle{IEEEtran}
\bibliography{refs}
}

\end{document}